%% file: paper.tex
\documentclass[runningheads]{llncs}
\usepackage[T1]{fontenc}
\usepackage{graphicx}
\usepackage{hyperref}
\usepackage{color}

\usepackage{amsmath}
\usepackage{amssymb}
\usepackage{bm}
\usepackage{tikz}
\usepackage{wrapfig}
\usepackage{environ}

\usepackage{changes}
\setauthormarkup{}
\definechangesauthor[name=me, color=blue]{TR}
\definechangesauthor[name=me, color=magenta]{JO}
\definechangesauthor[name=fab, color=teal]{FF}

\input{./macros}

\begin{document}
\title{How to peel fully convex digital sets}
\author{Fabien Feschet\inst{1}\orcidID{0000-0001-5178-0842} \and Jacques-Olivier Lachaud\inst{2}\orcidID{0000-0003-4236-2133}}
\authorrunning{F. Feschet and J.-O. Lachaud}
\institute{
    Universit{\'e} Clermont Auvergne, CNRS, ENSMSE, LIMOS, F-63000 Clermont-Ferrand, France\\
  \email{fabien.feschet@u-auvergne.fr}\\
  \and
  	Universit\'e Savoie Mont Blanc, CNRS, LAMA, F-73000 Chambéry, France\\
  \email{jacques-olivier.lachaud@univ-smb.fr}\\
}
\maketitle              %
\begin{abstract}
  Full convexity is an interesting alternative to classical digital
  convexity since it guarantees connectedness and even simple
  connectedness in digital spaces $\mathbb{Z}^d$, for any dimension
  $d$. This paper aims at giving a better understanding of the monotonicity
  properties of fully convex digital sets, since earlier works
  showed that the question was difficult for thin fully convex sets. To
  decipher the hierarchy of fully convex sets ordered by inclusion,
  we study how we can peel a fully convex set progressively while
  keeping its full convexity. We provide a characterization of
  peelable points and fast algorithms to identify them.  Furthermore
  we show that fully convex set can be peeled one point at a time till
  reduced to the empty set, similarly to digitally convex
  sets in the classical sense.  The peeling of a fully convex set can
  be seen as an analog to homotopic thinning processes, but with an
  additional geometric property.  \keywords{Digital convexity \and
    Homotopic thinning \and Digital geometry}
\end{abstract}

\section{Introduction}
\label{sec:intro}
\input{intro}

\section{Basic notions}
\label{sec:background}
\input{background}

\section{Sufficient conditions for sequential and parallel peelability}
\label{sec:thinning}
\input{thinning}

\section{Peelability}

\label{sec:peel}
\input{peelability}

\section{Conclusion and perspectives}
\label{sec:conclu}
\input{conclu}

\begin{credits}
  \subsubsection{\ackname}
  This work is partially supported by the French National Research
  Agency within the StableProxies project (ANR-22-CE46-0006). The
  authors also thank Gilles Bertrand for initiating discussions about
  thinning full convex sets.
\end{credits}

\bibliographystyle{splncs04}
\bibliography{biblio}

\appendix
\section{Peelable configurations in 2D}
\label{sec:appendix}
\input{appendix}

\end{document}

%% file: macros.tex
\newcommand{\R}{\ensuremath{\mathbb{R}}}
\newcommand{\Z}{\ensuremath{\mathbb{Z}}}
\let\C\relax \DeclareMathOperator{\C}{\ensuremath{\mathcal{C}^\text{$d$}}} %
\newcommand{\Ck}[1]{\ensuremath{\mathcal{C}^d_{#1}}}

\newcommand{\InterCk}[2]{\ensuremath{\bar{\mathcal{C}}^d_{#1} \lbrack #2 \rbrack}}

\newcommand{\InterC}[1]{\ensuremath{\bar{\mathcal{C}}^d \lbrack #1 \rbrack}}

\newcommand{\ve}[0]{\ensuremath{\mathbf{e}}}

\newcommand{\Conv}[1]{\mathrm{CvxH}\left( #1 \right)}
\newcommand{\Vertices}[1]{\ensuremath{\mathrm{Vtcs}\left(#1\right)}}

\newcommand{\Easy}[2]{\ensuremath{\mathrm{Easy}_{#1} \lbrack #2 \rbrack}}
\newcommand{\Hard}[2]{\ensuremath{\mathrm{Hard}_{#1} \lbrack #2 \rbrack}}

\newcommand{\RefFigure}[1]{Figure~\ref{#1}}
\newcommand{\RefLemma}[1]{Lemma~\ref{#1}}

\newcommand{\RefTheorem}[1]{Theorem~\ref{#1}}

\newcommand{\RefSection}[1]{Section~\ref{#1}}

\newcommand{\Star}[1]{\ensuremath{\mathrm{Star}\left(#1\right)}}

\newcommand{\Cl}[1]{\ensuremath{\mathrm{Cl}\left(#1\right)}}

\newcommand{\TCl}[1]{\ensuremath{\bar{#1}}}

\newcommand{\Le}{\ensuremath{\leqslant}}
\newcommand{\Ge}{\ensuremath{\geqslant}}

%% file: intro.tex
Convexity is a fundamental tool to analyse the
geometry of shapes in Euclidean spaces, and it is the basis of many other
fields like convex optimization. Our aim is to get a better
understanding of the properties of full convexity. Full convexity is a
recent alternative definition to digital convexity in the digital
lattice $\Z^d$ \cite{lachaud:2021-dgmm,lachaud:2022-jmiv}, which has
the considerable advantage of guaranteeing in arbitrary dimension the
digital connectedness of its elements and even the simple
connectedness of their geometric realization. This is in stark
contrast with classical digital convexity and to most other
adaptations of digital convexity (see \cite{kiselman:2022}
for a recent overview). Indeed, other
geometric definitions of digital convexity add the digital
connectedness as an external property
\cite{kim:1982-pr,kim:1982-tpami,kim:1982-tpami-3d,ronse:1989-tpami,Eckhardt2000},
but this trick does not give any insight at what is a true digital
analog to convexity and is not helpful in dimension greater or equal
to three. Besides, axiomatic definitions of digital convexity
\cite{webster2001cell,Llinares2002} or more functional versions
\cite{Murota2001,RoyStell2003,Kiselman2004} lose the Euclidean
geometric intuition of convexity and its relation to digital
planarity.

Full convexity is already showing a strong potential in digital
geometry analysis
\cite{lachaud:2021-dgmm,lachaud:2022-jmiv,feschet:2022-dgmm,feschet:2023-jmiv}:
local characterizations of convex and concave parts, unambiguous
notion of visibility, exact geodesics, tangent planes, computation
time comparable to classical convexity. Fully convex shapes also
includes standard digital primitives like digital planes and
balls. One can design a fully convex envelope operator idempotent on
fully convex set, which gives a proper definition of polyhedral
model. Full convexity also enjoys two new nice characterizations
\cite{feschet:2024-dgmm,feschet:2025-jmiv}.
However these works have highlighted a difficulty that is proper to
the fully convex envelope operator, and which is certainly related to
the particular connectedness properties of fully convex sets: the
envelope operator may not be increasing, and this may happen around
``thin'' digital sets.

In order to understand better this peculiar hierarchical structure of
fully convex sets and regular digital sets in-between, we propose in
this paper to study the \emph{peeling} of a fully convex set. More
precisely, we determine when the full convexity property is preserved
after removing one point, or even removing several points in
parallel. This can be seen as a generalization of the works of
Tarsissi \emph{et al.}
\cite{tarsissi2019convexity,tarsissi2022algorithms,tarsissi2023convexity},
which addresses the case of convex polyominoes (hence 4-connected
digital convex objects of $\Z^2$) and which focuses on the word
combinatorics formulation of these objects
\cite{brlek2009lyndon+}. There is also a strong relation with the
fruitful works on simple points and homotopic thinning
\cite{bertrand1994new,couprie2008new}, which has led to many
applications in 3D imaging \cite{pudney1998distance}, see also
\cite{bertrand2014powerful,couprie2016asymmetric} for parallel
homotopic thinning. In a sense, peeling is a more restrictive
homotopic thinning, since it preserves homotopy and convexity, and
that kind of property could be of interest in many skeletonization or
deformation algorithms.

The paper is organized as follows. \RefSection{sec:background}
presents the necessary background information, notations and
results. \RefSection{sec:thinning} provides sufficient conditions for
peelability, valid in a sequential or in a parallel
process. \RefSection{sec:peel} focuses on the characterization of
peelability. It is shown to be locally decidable in $\Z^2$, but no more
starting from $\Z^3$. We also show that any fully convex set is peelable
one point at a time till reduced to the empty set, confirming the
hierarchical (yet not arborescent) nature of the space of fully convex
sets. \RefSection{sec:conclu} summarizes the results and outlines some
perspectives to this work.

%% file: background.tex
We introduce here basic definitions and properties needed in the rest
of the paper (they can be found
in \cite{lachaud:2022-jmiv,feschet:2023-jmiv,feschet:2025-jmiv}).  In the sequel, $\C$
is the cubical cell complex induced by $\Z^d$. Its 0-dimensional cells
are identified to points of $\Z^d$. The set $\Ck{k}$ is the set of
open $k$-dimensional cells of $\C$.

The {\em (topological) boundary} $\partial Y$ of a subset $Y$ of
$\R^d$ is the set of points in its closure but not in its
interior. The star of a cell $\sigma$ in $\C$, denoted by
$\Star{\sigma}$, is the set of cells of $\C$ whose boundary contains
$\sigma$, plus the cell $\sigma$ itself.  The closure $\Cl{\sigma}$ of
$\sigma$ contains $\sigma$ and all the cells in its boundary. In this
paper, the cell boundary operator, also denoted by $\partial$, maps a
$k$-cell to all its proper faces, that is all its $k'$-cells, $0 \Le
k' < k$, and not only its $(k - 1)$-cells.

A subcomplex $K$ of $\C$ with $\Star{K} = K$ is \emph{open}, while
being \emph{closed} when $\Cl{K} = K$.  The \emph{body} of a subcomplex $K$,
i.e. the union of its cells in $\R^d$, is written $\|K\|$.

For any real subset $Y$ of $\R^d$, we denote by $\InterCk{k}{Y}$ the
set of $k$-cells whose topological closure intersects $Y$, i.e.
$\InterCk{k}{Y} = \lbrace c \in\Ck{k}, \TCl{c}\cap Y
\neq\emptyset\rbrace$, where $\TCl{c} = \|\Cl{c}\|$ for any cell
$c$. For any subset $Y \subset \R^d$, it is natural to define
$\Star{Y} := \InterC{Y} = \cup_{0\Le k \Le d} \InterCk{k}{Y}$. Last, the set $\Conv{Y}$ is the \emph{convex hull}
of $Y$ in $\R^d$.

\begin{definition}[Full convexity]
  A subset $X \subset \Z^d$ is {\em digitally $k$-convex} for $0 \Le k
  \Le d$ whenever
  \begin{align}\label{digital-k-convexity}
    \InterCk{k}{X} = \InterCk{k}{\Conv{X}}.
  \end{align}
  Subset $X$ is {\em fully (digitally) convex} if it is digitally
  $k$-convex for all $k, 0 \Le k \Le d$.
\end{definition}

The digital $0$-convexity is the classical digital convexity. The
following characterization will be useful:
\begin{lemma}[\protect{\cite[Lemma~4]{feschet:2022-dgmm}}]
	\label{lem-full-convexity-equ}
	A digital set $X$ is fully convex iff $\Star{X} = \Star{\Conv{X}}$.
\end{lemma}

Finite convex sets $Y$ in $\R^d$ are intended as $Y =
\Conv{\Vertices{Y}}$ where $\Vertices{Y}$ is the set of extreme points
of $Y$, also called \emph{vertices} in this paper.
By definition, $\Vertices{Y}\subset Y$. Moreover, any one of the
vertices of $Y$ cannot be written as a convex linear combination of
other points of $Y$. Since we are considering finite digital sets, the
convex set $\Conv{Z}$ is always a bounded polytope of $\R^d$ for all
finite subsets $Z$ of $\Z^d$. A vertex $z$ of a convex set $\Conv{Y}$
is such that $z$ does not belong to any open segment $ \rbrack
u,w \lbrack$ with $u,w \in \Conv{Y}$. We can specialize this property
for digital sets. Let us denote by $z[\ve]$ the line with direction
$\ve$ passing through $z$. It is the union of two rays $z^{+}[\ve]$
and $z^{-}[\ve]$ which are both semi-infinite in the direction of
$+\ve$ and resp. $-\ve$, and which contains $z$. We then have the
following property.

\begin{lemma}
  For $X\subset\Z^d$, $z\in\Vertices{\Conv{X}}$ and $\varepsilon\in\{+,-\}$,
  $z^{\varepsilon}[\ve]\cap X\neq\emptyset$ $\Rightarrow$ $z^{-\varepsilon}[\ve]\cap X=\emptyset$.
\end{lemma}
\begin{proof}
	Indeed, if both intersections are not empty then $z$ is a convex combination of others points of $X$ and cannot be a vertex of $\Conv{X}$.\qed
\end{proof}

Fully convex sets are stable with respect to orthogonal projections
along axes in $\R^d$. Let $I_d=\{1,\ldots,d\}$ and for $j \in I_d$, let
$\pi_{j}$ denote the orthogonal projector associated to the $j$-th
axis, which consists in omitting the $j$-th coordinates for all points
of $\R^d$. The image of a point in $\Z^d$ by
any axis projector is a point in $\Z^{d-1}$. Let us recall the definition of $P$-convexity \cite{feschet:2025-jmiv} and its relation with full convexity.

\begin{definition}[$P$-convexity]
  Let $X \subset \Z^d$ be a digital set. The set $X$ is {\em
    $P$-convex} if and only if $X$ is digitally 0-convex
  (i.e. $\Conv{X} \cap \Z^d = X$) and when $d > 1$, for any $j \in I_d$,
  $\pi_j(X)$ is $P$-convex in $\Z^{d-1}$.
\end{definition}

\begin{theorem}[\cite{feschet:2025-jmiv}, theorem 5.]
  \label{thm-p-convex-quiv-full-convex}
  For arbitrary dimension $d \Ge 1$, for any $X \subset \Z^d$, $X$ is
  fully convex if and only if $X$ is $P$-convex.
\end{theorem}

The property below recalls that the vertices of an axis projected convex
set are to be found among the projection of the vertices of the
original convex set.

\begin{proposition}\label{prop:projectedVtcs}
  Let $X$ be a 0-convex set. Then for each direction $i \in I_d$,
  $\pi_i(X)$ is 0-convex and $\forall z\in\Vertices{\Conv{\pi_i(X)}}$, $\exists z'\in\Vertices{\Conv{X}}$, $z = \pi_i(z')$.
\end{proposition}

%% file: thinning.tex
We present in this part two possible ways to peel fully convex sets
until reaching a unique point. A first way considers the vertices of
the convex hull of the digital set and their possible sequential and
parallel elimination. A second way is to peel along the bounding box
slice per slice.

The following lemma tells that we can peel a digitally convex set by
its vertices. This does not however guarantee to maintain
full convexity.

\begin{lemma}
  \label{lem-peeling-0-convex-set}
  Let $X \subset \Z^d$ be a $0$-convex set, i.e. $\Conv{X}\cap
  \Z^d=X$, and $V:=\Vertices{\Conv{X}}$. Then for any subset of
  vertices $Z \subset V$, $X \setminus Z$ is 0-convex.
\end{lemma}
\begin{proof}
  Let $X':=X \setminus Z$. We have $\Conv{X'}\cap\Z^d \supset X'$
  since $\Conv{\cdot}$ is increasing and $X' \subset \Z^d$. Let us
  thus show that $\Conv{X'}\cap\Z^d \subset X'$. This is true when
  $X'=\emptyset$ so we can suppose that $X'\neq\emptyset$, which
  implies that $\Conv{X'}\cap\Z^d\neq\emptyset$.

  We reason by contradiction: let suppose that $\exists x \in
  \Conv{X'} \cap \Z^d$, while $x \not\in X'$. We have also $x \in
  \Conv{X} \cap \Z^d$ (since $\Conv{\cdot}$ is increasing and $X'
  \subset X$). The fact that $X$ is 0-convex implies $x \in X$. It
  follows that $x \in X \setminus X'=Z \subset V$. Since $x\notin X'$
  and since $\Vertices{{\Conv{X'}}}\subset X'$,
  $x\notin\Vertices{{\Conv{X'}}}$. So, $x \in \Conv{X'}
  \cap \Z^d$ implies that $x$ is a convex combination of
  $\Vertices{{\Conv{X'}}}$. But $X'\subset X$ implies
  that $x$ is a convex combination of
  $\Vertices{{\Conv{X}}}$, which is a contradiction with $x\in V =\Vertices{\Conv{X}}$.\qed
\end{proof}
It is necessary to peel convex hull vertices if we peel $X$ one point
at a time.
\begin{lemma}
  \label{lem-peeling-seq-iff}
  Let $X \subset \Z^d$ be a 0-convex set. Let $z \in X$. Then $X
  \setminus \{z\}$ is 0-convex if and only $z$ is a vertex of
  $\Conv{X}$.
\end{lemma}
\begin{proof}
  $(\Leftarrow)$ by using \RefLemma{lem-peeling-0-convex-set}.

  \begin{sloppypar}
    $(\Rightarrow)$ Assume $X\setminus\lbrace z\rbrace$ is 0-convex,
    i.e. $\Conv{X\setminus\lbrace z\rbrace}\cap\Z^d = X\setminus\lbrace
    z\rbrace$. Hence, $z\notin \Conv{X\setminus\lbrace z\rbrace}$
    because $z\in\Z^d$. But since $\Conv{X} = \Conv{\Vertices{X}}$,
    $\Conv{X\setminus\lbrace z\rbrace}\subsetneq\Conv{X}$ implies that
    $\Vertices{{\Conv{X\setminus\lbrace
          z\rbrace}}}\neq\Vertices{{\Conv{X}}}$. So at
    least one point of $\Vertices{{\Conv{X}}}$ has been
    removed. Only $z$ was removed so $z$ is a vertex of $\Conv{X}$.
    \qed
  \end{sloppypar}
\end{proof}

We show below that a local test is sufficient to determine if removing a
vertex of a fully convex set keeps the full convexity property.
\begin{theorem}
  \label{thm-par-peeling-fully-convex}
  Let $X$ be a non empty finite fully convex set of $\Z^d$ and let $Z
  \subset \Vertices{\Conv{X}}$ be a subset of its convex hull
  vertices. Assume that, for every direction $i \in I_d$,
  for every $z \in Z$, we have $\{z -\ve_i, z +\ve_i\} \cap (X
  \setminus Z) \neq \emptyset$. Then $X \setminus Z$ is fully convex.
\end{theorem}
\begin{proof}
  We prove that $X \setminus Z$ is $P$-convex. First of all, the
  result is trivial if $Z$ is empty. Otherwise, the first requirement
  for $P$-convexity is that $X \setminus Z$ is 0-convex, which is a
  consequence of \RefLemma{lem-peeling-0-convex-set}. This is enough
  to conclude if $d=1$.

  Otherwise, if $d \Ge 2$, we have to prove that for every direction
  $i \in I_d$, $\pi_i(X \setminus Z)$ is $P$-convex. Let
  $X':=X \setminus Z$. Let $Y:=\pi_i(Z)$. We consider the two subsets
  of $X$ and $X'$ that projects exactly onto $Y$,
  i.e. $C:=\pi_i^{-1}(Y) \cap X$ and $C':=\pi_i^{-1}(Y) \cap X'$. It
  is clear that $C$ is non empty since $Z$ is included in $C$, and
  that $\pi_i(C)=Y$ ($Z \subset C$, $\pi_i(Z) \subset \pi_i(C) \subset
  Y$ and $\pi_i(Z) = Y$ by definition).

  Looking at $C'$, for any $y \in Y$, there is some $z \in Z$ with
  $z=\pi_i(z)$. By hypothesis $\{z -\ve_i, z +\ve_i\} \cap (X
  \setminus Z) \neq \emptyset$. Let $z' \in \{z -\ve_i, z +\ve_i\}
  \cap (X \setminus Z)$. Clearly $\pi_i(z')=\pi_i(z \pm
  \ve_i)=y$. Furthermore $z' \in X \setminus Z = X'$ and projects onto
  $Y$. So $z' \in C'$. Since for any $y \in Y$, we have found some $z'
  \in C'$ that projects onto it, we have $\pi_i(C') \supset Y$. And
  obviously $\pi_i(C')=\pi_i(\pi_i^{-1}(Y) \cap X') \subset Y$. This
  proves $\pi_i(C')=Y$.

  The sets $\pi_i(X \setminus C)$ and $\pi_i(C)$ are disjoint since
  $C$ contains all the elements of $X$ that projects onto $Y$. The
  same is true for $\pi_i(X' \setminus C')$ and $\pi_i(C')$. The
  equality $X \setminus C = X' \setminus C'$ holds (using $\pi_i(x)
  \not\in Y=\pi_i(Z)$ implies $x \not\in Z$) since:
  \begin{align*}
    x \in X' \setminus C' \Leftrightarrow x \in X, x \not\in Z, \pi_i(x) \not\in Y
    \Leftrightarrow x \in X, \pi_i(x) \not\in Y
    \Leftrightarrow x \in X \setminus C.
  \end{align*}
  \begin{align*}
    \text{Then\quad}\pi_i(X') & = \pi_i(X' \setminus C') \cup \pi_i(C') && \text{(disjoint union and $C' \subset X'$)}\\
    & = \pi_i(X' \setminus C') \cup \pi_i(C) && \text{(since $\pi_i(C')=\pi_i(C)=Y$)}\\
    & = \pi_i(X \setminus C)  \cup \pi_i(C) && \text{($X \setminus C = X' \setminus C'$)}\\
    & = \pi_i(X). && \text{(since $C \subset X$)}
  \end{align*}
  Since $X$ was fully convex by hypothesis, it is also $P$-convex, so
  its projection $\pi_i(X)$ is $P$-convex. So $\pi_i(X')$ is
  $P$-convex and \RefTheorem{thm-p-convex-quiv-full-convex} concludes.\qed
\end{proof}

The previous theorem induces a very simple algorithm to thin a fully
convex set $X$ sequentially up to a given constraint set $Y$: (i)
extract vertices of the convex hull of $X$, (ii) check if one vertex
that is not in $Y$ satisfies the theorem hypothesis, (iii) remove it
from $X$ and loop to (i) till no such vertex can be found. However it
requires to compute or update the convex hull of $X$ at each step. The
same theorem shows sufficient conditions under which we can peel $X$
in parallel. The convex hull must still be updated but for generally
less iterations.

It is also quite easy to remove points on the bounding box of a fully
convex set $X$.  We recall \cite[Lemma~5]{lachaud:2022-jmiv}: if $X$
is fully convex and $Y \subset \R^d$ is \emph{stable}, then $X \cap Y$
is fully convex. Now \cite[Lemma~6]{lachaud:2022-jmiv} states that any
half-space of integer intercept and axis normal vector is stable, and
any intersection of stable sets is itself stable. The result below is
straightforward by choosing $P$ as the half-space just
touching one side of the
bounding box.

\begin{proposition}\label{prop:bb}
  Let $X \subset \Z^d$ fully convex, non-empty, finite. Let
  $(l,u)$ be the lowest and highest points of its axis-aligned
  bounding box.
  If $P$ is a $d-1$-hyperplane orthogonal to an axis, with $l$ or $u$ in $P$, then $X \setminus P$ is fully convex.
\end{proposition}

We can thus remove any one of the side of $X$ and yet keep its
full convexity.

%% file: peelability.tex
Previous peelability results apply only in some circumstances
and do not show that one can peel a set one point at a time.
To provide a more precise answer on the
deletion of points in fully convex sets while maintaining the full
convexity, we therefore provide a definition of what peelable really
means.

\begin{definition}\label{peelable}
For a fully convex set $X$ and some $z \in X$, we say that $z$ is
\emph{peelable in $X$} when $z \in \Vertices{\Conv{X}}$ and, when $d
\Ge 2$, for every direction $i \in I_d$, either (i) $z
-\ve_i \in X$ or $z+\ve_i \in X$, or (ii) $\pi_i(z)$ is peelable in
$\pi_i(X)$.
\end{definition}

When $d=1$, a point $z$ is peelable as soon as it is a vertex of
$\Conv{X}$. Let us also remark that when using projections, a peelable point $z$ must but projected onto a vertex of the convex hull of the projection such that the condition $z \in \Vertices{\Conv{X}}$ must be checked again. We now explain why peelability is not a local notion. To
do this, we first prove that, for $d \Le 2$, peelability of a vertex
$z$ is indeed local, and then provides a counter example in $d \Ge 3$
of local decidability of peelability for $z$.

For some $z \in \Z^d$, we consider $\mathcal{N}_X(z) =
\Cl{\Star{z}}\cap X$ the neighborhood of $z$ in $X$. Let us show now
that peelability is locally decidable in 2d (see appendix
\ref{sec:appendix} for the complete list of local configurations and a
fast implementation).

\begin{lemma}\label{lem:local2d}
  Let  $X\subset\Z^2$, finite and full convex. Then point $z
  \in \Vertices{\Conv{X}}$ is peelable in $X$ if and only if $z$ is
  peelable in $\mathcal{N}_X(z)$.
\end{lemma}
\begin{proof}
 The cases depend on the number of intersections between
 $X$ and rays $z[.]$ along directions $\ve_1$ and
 $\ve_2$ (rotations excluded), $N=\sharp\{X\cap\left(z[\ve_1]\cup z[\ve_2]\right)\}$.

 \centerline{\hfill\scalebox{0.75}{\input{2dCardNone.tikz}}\hfill
     \scalebox{0.75}{\input{2dCardOnea.tikz}}$\qquad$
     \scalebox{0.75}{\input{2dCardOneb.tikz}}\hfill
     \scalebox{0.75}{\input{2dCardTwo.tikz}}\hfill}

 \textbf{($\bm{N=0}$)} $z$ is peelable in $X$ and in $\mathcal{N}_X(z)$
 because $z$ is projected at a vertex for any direction $\ve_1$ and
 $\ve_2$. \textbf{($\bm{N=1}$ left)} $z$ is peelable in $X$ and in
 $\mathcal{N}_X(z)$ because $z$ is only projected using $\ve_2$ and it
 is a vertex both in projection for $X$ and for
 $\mathcal{N}_X(z)$. \textbf{($\bm{N=1}$ right)} $z$ is not peelable along
 direction $\ve_2$ because it is covered by an edge of $\Conv{X}$, but
 it is also not peelable in $\mathcal{N}_X(z)$ along $\ve_2$ because
 it is not a vertex in projection. Note that point $e = z
 -\ve_1-\ve_2$ must belong to $X$ as $X$ is fully convex: otherwise
 $\Conv{X}$ would intersect a 1d cell $c$ not touching $X$ so $c
 \notin \Star{X}$, but $c \in \Star{\Conv{X}}$ which is a
 contradiction to \RefLemma{lem-full-convexity-equ}. \textbf{($\bm{N=2}$)}
 Point $z$ is peelable in both $X$ and $\mathcal{N}_X(z)$ because it
 has a neighbor using both direction $\ve_1$ and $\ve_2$.\qed
\end{proof}

Starting at $d=3$, peelability cannot be locally decided as explicited
in \RefFigure{fig-counter-example-3d}. It is quite easy to generalize
this example so that it requires arbitrary large neighborhoods to
decide the peelability.

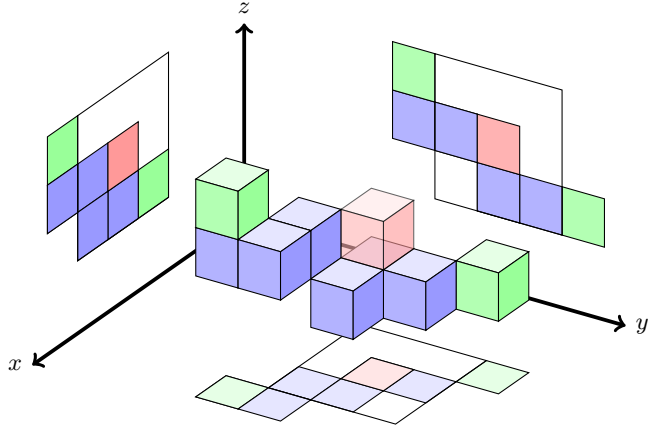
\begin{figure}[t]
  \begin{center}
    \input{counter-example.tikz}
    \vspace{-5mm}
  \end{center}
  \caption{\label{fig-counter-example-3d} Counter-example to local
    decidability of peeling in 3D. The fully convex set $X$ is
    composed of the blue, green and red points (represented as
    cubes). The red point is the vertex $p$. As one can see, $p$ is locally
    peelable along each projection (see \RefFigure{fig-cfg-2d} in
    appendix). However it is not a vertex of the convex hull of the
    projection of $X$ along axis $x$ into the plane $(yz)$, since it
    lies in the middle of the edge formed by the projection of the two
    green points. So $X \setminus \{p\}$ is not $P$-convex, hence not
    full convex. }
\end{figure}

We remark first that peelability imposes no specific ordering on projections.
\begin{lemma}\label{lem:projPermutations}
  For $J=(i_1,\cdots,i_p)$, let $\pi_J := \bigcirc_{k=1..p} \pi_{i_k}$
  be the projection composing every axis projections in $J$, and
  let $\sigma_J$ be a permutation of $J$. Then $\pi_{\sigma_J} =
  \pi_J$.
\end{lemma}
\begin{proof}
  It is sufficient to show it for two projections with directions
  $i\neq j$. Let us pick $z = (z_1,\cdots,z_d)\in\Z^d$. $\pi_k(z) =
  z_{\overline{k}}$ where $\overline{k}$ means omission of the
  coordinate of index $k$ that is the set $I_d\setminus
  \{k\}$. So $\left(\pi_i \bigcirc \pi_j\right)(z) =
  z'_{\overline{i}}$ where $z' = z_{\overline{j}}$. Since $i\neq j$,
  $z'_{\overline{i}} = z_{\overline{i,j}}$. Thus, $\left(\pi_i
  \bigcirc \pi_j\right)(z) = z_{\overline{i,j}} = z_{\overline{j,i}} =
  \left(\pi_j \bigcirc \pi_i\right)(z)$.\qed
\end{proof}

We start with an initial fully convex set $X$ in $\Z^d$. A point
$z\in\Z^d$ is peelable only if $z\in\Vertices{\Conv{X}}$. If it is the
case there exists a set of directions, called \emph{easy directions},
defined by $\Easy{X}{z} := \{i \in I_d : \exists
\varepsilon_i\in\{+1,-1\}, z+ \varepsilon_i\ve_i\in X \}$. So to check
peelability, projections will only be tested for \emph{hard
directions}, defined as $\Hard{X}{z} := I_d \setminus \Easy{X}{z}$. If
$\Hard{X}{z} = \emptyset$ then $z$ is a peelable point and all such
points are peelable.

Hard directions are stable in projection since emptyness of $z[\ve_i]$
is stable under $\pi_i$ and obviously emptyness in projection implies
emptyness in $\pi^{-1}$. So the following property holds:
\begin{lemma}\label{lem:hardProj}
    Suppose $X$ is a non empty fully convex set in $\Z^d$, and let $z\in\Vertices{\Conv{X}}$. Assume $i \neq j$ are two different directions of $I_d$ with $\pi_i(z)\in\Vertices{\Conv{\pi_i(X)}}$. We have $j\in \Hard{\pi_i(X)}{\pi_i(z)}$ $\Leftrightarrow$ $j\in \Hard{X}{z}$.
\end{lemma}

Suppose that $\Hard{X}{z} \neq \emptyset$. Let us remark that the
set $\Easy{X}{z}$ is preserved in the sense that for any
projection $i\in \Hard{X}{z}$ and any $\varepsilon_i\in\{+1,-1\}$,
it holds that $z+ \varepsilon_i\ve_k \in X \Rightarrow
\pi_i(z)+\varepsilon_i\pi_i(\ve_k)\in\pi_i(X)$ for $k\in
\Easy{X}{z}$. In other words, when checking the peelability of $z$
and taking projections in $\Hard{X}{z}$, we preserve the fact that
directions in $\Easy{X}{z}$ do not have to be checked. Hence, we have to study what happened when peelability conditions
are not verified for hard directions.

Point $z \in X$ is not peelable if and only if either
$z\notin\Vertices{\Conv{X}}$ or $z\in\Vertices{\Conv{X}}$ and there
exists a direction $i$ such that: (i) $z-\ve_i \notin X$ and $z+\ve_i
\notin X$ and (ii) $\pi_i(z)$ is not peelable in $\pi_i(X)$. So it is
clear that the condition stopping the iteration of the projections is
that $z\notin\Vertices{\Conv{X}}$ in some
projection. The following lemma explains this case geometrically.

\begin{lemma}\label{lem:domination}
  Let $X$ fully convex and let $z\in\Vertices{\Conv{X}}$ and a
  direction $i \in \Hard{X}{z}$. Then
  $\pi_i(z)\notin\Vertices{\pi_i\left(\Conv{X}\right)} \Leftrightarrow
  \exists F^{\varepsilon}$ face of $\Conv{X}$ such that
  $z^{\varepsilon}[\ve_i]\cap\text{\emph{relint}}(F^{\varepsilon})\neq\emptyset$
  for some $\varepsilon\in\{+1,-1\}$.
\end{lemma}
\begin{proof}
  (Preamble) Since $z \in X$, $z+\ve_i\notin X$ and $z-\ve_i\notin X$,
  we have that $z + k\ve_i\notin X$ for any $k\in\Z \setminus \{0\}$
  since $X$ is fully convex ($0$-convexity suffices).

  $(\Rightarrow)$ If
  $\pi_i(z)\notin\Vertices{\pi_i\left(\Conv{X}\right)}$ then there
  exists a segment $[u,v]$ in $\pi_i\left(\Conv{X}\right)$ such that
  $\pi_i(z)\in \rbrack u,v \lbrack$. Considering a segment of the form
  $\rbrack u',v' \lbrack$ with $u'\in\pi_i^{-1}(u)$, $v'\in\pi_i^{-1}(v)$
  and $u',v'\in\Conv{X}$, we note that $z[\ve_i]\cap \rbrack u',v' \lbrack
  \neq\emptyset$ and that $\pi_i\left(z[\ve_i]\cap \rbrack u',v' \lbrack
  \right) = \pi_i(z)$ for all possible $u'$ and $v'$.

  Suppose that $z[\ve_i]\cap\Conv{X} = \{z\}$. For any $u'$ and $v'$,
  there is a point of $\rbrack u',v' \lbrack$ in $z[\ve_i]$. By
  hypothesis this point is $z$. Hence $z$ is a convex combination of
  $u'$ and $v'$ which is a contradiction. So, $z[\ve_i]\cap\Conv{X}$ is
  a segment $S$ containing $z$ as a vertex. Consider the other vertex
  $v$ of $S$ and a facet $F^{\varepsilon}$ containing it. Point $v$
  cannot be a vertex of $F^{\varepsilon}$ since this would be a point
  in $\Z^d$ other than $z$ and belonging to $z[\ve_i]$ which is
  impossible by the preamble. Hence
  $v\in\text{relint}(F^{\varepsilon})$. We get that
  $z^{\varepsilon}[\ve_i]\cap\text{relint}(F^{\varepsilon})\neq\emptyset$.

  $(\Leftarrow)$ The point $v$ in
  $z^{\varepsilon}[\ve_i]\cap\text{relint}(F^{\varepsilon})$ is not in
  $\Z^d$ with $\varepsilon\in\{+1,-1\}$. So the projections of any
  vertex of the face $F^{\varepsilon}$ are different from
  $\pi_i(z)$. Hence $\pi_i(z)$ belongs to the convex hull of the
  projections of $\Vertices{\mathcal{F^{\varepsilon}}}$ because convex
  combinations are preserved by the linear operator $\pi_i$. This
  implies that
  $\pi_i(z)\notin\Vertices{\pi_i\left(\Conv{X}\right)}$.\qed
\end{proof}

\RefLemma{lem:domination} tells us that when the projection of a vertex $z$ is covered by the projection of a face then $z$ would not be peelable. This indicates that points on the boundary of $X$ are the points which are the most difficult to cover using projections. We now give the main result for the usage of peelability which legitimates its definition.

\begin{theorem}
  \label{thm-non-local-seq-peeling-fully-convex}
  Let $X$ be a non empty finite fully convex set of $\Z^d$. Then, $z$
  is a peelable point of $X$ $\Leftrightarrow$ $X \setminus \{z\}$ is
  fully convex.
\end{theorem}
\begin{proof}
  $(\Rightarrow)$ If $z$ is a peelable point of $X$, then
  $z\in\Vertices{\Conv{X}}$. We know that vertex removal preserves
  $0$-convexity (\RefTheorem{thm-par-peeling-fully-convex}) when
  condition (i) of peelability is verified for all directions.
  Else we
  can recursively use the peelability along projections, as in the
  definition of $P$-convexity. Noticing that $0$-convexity is
  preserved by projections of full convex sets - while it is not true in general for arbitrary $0$-convex sets - and using
  \RefTheorem{thm-p-convex-quiv-full-convex}, it is clear that
  removing peelable points preserves full convexity.

  $(\Leftarrow)$ Suppose that both $X$ and $X \setminus\{z\}$ are full
  convex. We cannot have $\Conv{X} = \Conv{X\setminus\{z\}}$ because
  it would imply that $\Conv{X}\cap\Z^d =
  \Conv{X\setminus\{z\}}\cap\Z^d$ which is impossible because $z$ is
  in the former set but not in the later one. So since $\Conv{X} \neq
  \Conv{X\setminus\{z\}}$ then $\exists z_0\in\Vertices{\Conv{X}}$
  such that $z_0\notin\Vertices{\Conv{X\setminus\{z\}}}$. So in
  particular $z_0\notin X\setminus\{z\}$. But $X
  \setminus \left(X\setminus\{z\}\right) = \{z\}$. Hence
  $z = z_0$ and $z\in\Vertices{\Conv{X}}$.

  If $\Easy{X}{z}=I_d$ then $z$ is peelable. Otherwise $\Hard{X}{z}$ is
    not empty and let $i \in \Hard{X}{z}$.
  We consider $\mathcal{Z}_i = z[\ve_i]\cap \Conv{X}$. Assume
  $\mathcal{Z}_i\neq \{z\}$. So we consider the farthest face
  $F^{\varepsilon}$ such that $z[\ve_i]\cap
  F^{\varepsilon}\neq\emptyset$ using direction $\varepsilon\ve_i$
  from $z$. By construction $F^{\varepsilon}$ if also a face of
  $\Conv{X\setminus\{z\}}$. But $F^{\varepsilon}$ separates $z$ from
  $z+\varepsilon e_i\notin X$. Consider now the 1d cell $c\in\C$ whose
  realization is $\|c\| = \rbrack z ; z + \varepsilon\ve_i
  \lbrack$. It is clear that $\|c\| \cap
  F^{\varepsilon}\neq\emptyset$. Since $F^{\varepsilon}$ is a face of
  $\Conv{X\setminus\{z\}}$, we have $\|c\| \cap
  \Conv{X\setminus\{z\}}\neq\emptyset$. So since $X\setminus\{z\}$ is
  full convex, we must have that
  $\Vertices{\|c\|}\cap\left(X\setminus\{z\}\right)\neq\emptyset$. This
  is a contradiction since neither points are in $X\setminus\{z\}$.

  So, for any $i \in \Hard{X}{z}$, we must have $\mathcal{Z}_i = z[\ve_i]\cap
  \Conv{X} = \{z\}$. But using \RefLemma{lem:domination}, this implies
  that $\pi_i(z)\in\Vertices{\pi_i(\Conv{X})}$ for direction
  $i$. Hence $z$ is a peelable point.\qed
\end{proof}

Thanks to \RefTheorem{thm-non-local-seq-peeling-fully-convex}, we know
that peelable points are necessary and sufficient to peel fully convex
sets while maintaining full convexity. However, we have to prove that
such points exist for every non empty fully convex set. This is done
by the following theorem (see \RefFigure{fig-bbox} for an
illustration).

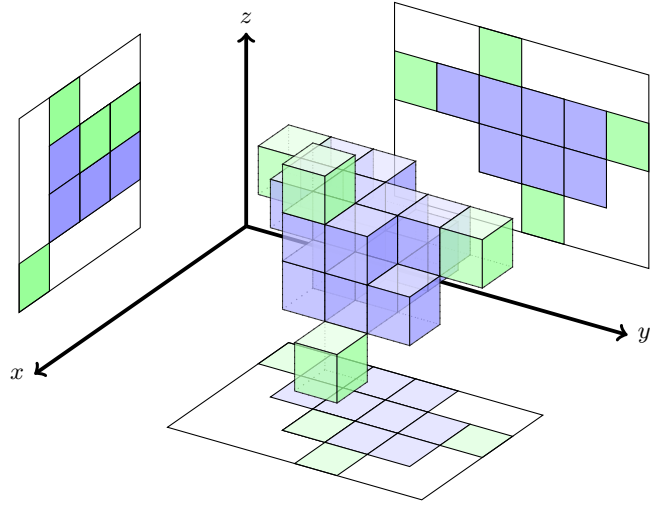
\begin{figure}[t]
  \begin{center}
    \input{bbox-example.tikz}
  \end{center}
  \caption{\label{fig-bbox} Any facet of the bounding box of a fully convex set $X$ contains a peelable point (displayed in green).}
\end{figure}

\begin{theorem}\label{thm:BB}
    Suppose $X$ is a non empty fully convex set in $\Z^d$. Then, any facet of the bounding box of $X$ contains a peelable point.
\end{theorem}
\begin{proof}
    We are going to prove the theorem by recursion on the
    dimension. To do this, we need to remark that for any direction
    $i\in I_d$ the bounding box facets for the projection $\pi_i(X)$
    are precisely the projections of the facets of the bounding box of
    $X$. Thus, any facet $f$ of the bounding box of
    $\pi_i(X)$ can be written as $\pi_i(F)$ where $F$ is a facet of
    the bounding box of $X$. We can made this fact precise simply by
    looking at the facets of the bounding box of $X$ whose normal
    vectors are the $\pm\ve_j$ for $j\neq i$. Indeed, these are
    precisely the facets whose projections define the bounding box of
    $\pi_i(X)$.

    Let us now consider the starting case: $d=1$. As previously seen,
    a point $z$ is peelable for $d=1$ if and only if $z$ is a vertex
    of $X$. So $X$ has one peelable point if $X = \{z\}$ and two
    peelable points when $X$ has two points or more. So the property
    is true for $d=1$.

    Let us suppose that for any dimension less than $d$, the property
    is true. We consider a direction $i\in I_d$. Let us denote by
    $X_i:=\pi_i(X)$. Let us consider a facet $f$ of the bounding box
    of $X_i$. As recalled previously, $f = \pi_i(F)$ where $F$ is a
    facet of the bounding box of $X$. By hypothesis, the dimension of
    $X_i$ is strictly lower than the dimension of $X$. Thus applying
    the hypothesis since $\pi_i(X)$ is fully convex, we get that in
    $f$, there exists a peelable point $z_f$. Using
    \RefLemma{prop:projectedVtcs}, we now that, $z_f$ being a vertex
    of $\Conv{X_i}$ in $f$, it is the projection of a vertex $z_F$ of
    $\Conv{X}$. By construction $z_F\in F$. But since we know that
    hard directions are stable by projection
    (\RefLemma{lem:hardProj}), we know that all directions which are
    hard in projections behave well because $z_f$ is peelable. The
    peelability of $z_F$ is thus solely govern by the remaining
    direction $i$. There are two cases. If $i\in \Easy{X}{z_F}$ then
    by definition $z_F$ is peelable. If $i\in \Hard{X}{z_F}$, then we
    note that $z_f\in\Vertices{X_i}$ and thus direction $i$ behave
    well in the sense that $\pi_i(z_F) = z_f$ is a vertex of
    $X_i$. Hence again $z_F$ is peelable. Overall we obtain that the
    face $F$ has a peelable point $z_F$. But, since the all the facets
    of the bounding box of $X$ are constructible by the facets of the
    bounding box of $X_i$ (by considering all the
      directions $i \in I_d$), we obtain that the property of true
    for dimension $d$.  \qed
\end{proof}

%% file: 2dCardNone.tikz
\begin{tikzpicture}[x=0.4cm,y=0.4cm]
	\draw[help lines,step=0.4cm] (-2,-2) grid (2,2);
	\draw[draw,fill=blue,color=blue,fill opacity=0.2] (0,0) -- (0.6,-2) -- (2,-2) -- (2,-0.6) -- (0,0);
	\foreach \x/\y in {0/0,1/-1} {
		\draw[fill=black,color=black] (\x,\y) circle (0.5mm);
	}
	\foreach \x/\y in {-1/0,1/0,0/-1,0/1,-1/1,1/1,-1/-1} {
		\draw[color=black] (\x-0.2,\y-0.2) rectangle ++(0.4,0.4);
	}
	\node at (1.5,-1.5) {$X$};
	\node at (0,-3) {$N=0$};
\end{tikzpicture}

%% file: 2dCardOnea.tikz
\begin{tikzpicture}[x=0.4cm,y=0.4cm]
	\draw[help lines,step=0.4cm] (-2,-2) grid (2,2);
	\draw[draw,fill=blue,color=blue,fill opacity=0.2] (0,0) -- (0.6,2) -- (2,2) -- (2,-2) -- (2,-2) -- (0.6,-2) -- (0,0);
	\foreach \x/\y in {0/0,1/0} {
		\draw[fill=black,color=black] (\x,\y) circle (0.5mm);
	}
	\foreach \x/\y in {-1/0,0/-1,0/1,-1/1,-1/-1} {
		\draw[color=black] (\x-0.2,\y-0.2) rectangle ++(0.4,0.4);
	}
	\foreach \x/\y in {1/1,1/-1} {
		\node at (\x,\y) {?};
	}
	\node at (1.5,0.5) {$X$};
	\node at (0,-3) {$N=1$};
\end{tikzpicture}

%% file: 2dCardOneb.tikz
\begin{tikzpicture}[x=0.4cm,y=0.4cm]
	\draw[help lines,step=0.4cm] (-2,-2) grid (2,2);
	\draw[draw,fill=blue,color=blue,fill opacity=0.2] (0,0) -- (-2,-0.7) -- (-2,-1.7) -- (2,0) -- (2,0.5) -- (0,0);
	\foreach \x/\y in {0/0,1/0,-1/-1} {
		\draw[fill=black,color=black] (\x,\y) circle (0.5mm);
	}
	\foreach \x/\y in {-1/0,0/-1,0/1,-1/1} {
		\draw[color=black] (\x-0.2,\y-0.2) rectangle ++(0.4,0.4);
	}
	\foreach \x/\y in {1/1} {
		\node at (\x,\y) {?};
	}
	\node at (-1.5,-1) {$X$};
	\node at (0,-3) {$N=1$};
\end{tikzpicture}

%% file: 2dCardTwo.tikz
\begin{tikzpicture}[x=0.4cm,y=0.4cm]
	\draw[help lines,step=0.4cm] (-2,-2) grid (2,2);
	\foreach \x/\y in {0/0,1/0,0/-1} {
		\draw[fill=black,color=black] (\x,\y) circle (0.5mm);
	}
	\foreach \x/\y in {-1/0,0/1,-1/1} {
		\draw[color=black] (\x-0.2,\y-0.2) rectangle ++(0.4,0.4);
	}
	\foreach \x/\y in {1/1,-1/-1} {
		\node at (\x,\y) {?};
	}
	\node at (0,-3) {$N=2$};
\end{tikzpicture}

%% file: counter-example.tikz
\tikzset{math3d/.style={x= {(-0.5cm,-0.35cm)}, z={(0cm,0.8cm)},y={(0.7cm,-0.2cm)}}}
\NewDocumentCommand{\DrawVoxel}{m m m m O{1.0}}{%
  \draw[fill=#4!10!white,draw=black,thin,opacity=#5] (#1-0.5,#2-0.5,#3+0.5) -- (#1+0.5,#2-0.5,#3+0.5) -- (#1+0.5,#2+0.5,#3+0.5) -- (#1-0.5,#2+0.5,#3+0.5) -- cycle;
  \draw[fill=#4!30!white,draw=black,thin,opacity=#5] (#1+0.5,#2-0.5,#3-0.5) -- (#1+0.5,#2-0.5,#3+0.5) -- (#1+0.5,#2+0.5,#3+0.5) -- (#1+0.5,#2+0.5,#3-0.5) -- cycle;
  \draw[fill=#4!40!white,draw=black,thin,opacity=#5] (#1-0.5,#2+0.5,#3-0.5) -- (#1+0.5,#2+0.5,#3-0.5) -- (#1+0.5,#2+0.5,#3+0.5) -- (#1-0.5,#2+0.5,#3+0.5) -- cycle;
}
\NewDocumentCommand{\DrawPoint}{m m m m O{1.0}}{%
  \draw[fill=#4,draw=black,thin,opacity=#5] (#1,#2,#3) node {\color{#4!50!white}$\bullet$};
}
  
\def\ProjX{-2.0}
\def\ProjY{-4.0}
\def\ProjZ{-2.0}
\NewDocumentCommand{\ProjectVoxel}{m m m m O{\ProjX} O{\ProjY} O{\ProjZ}}{%
  \draw[fill=#4!10!white,draw=black,thin] (#1-0.5,#2-0.5,#7) -- (#1+0.5,#2-0.5,#7) -- (#1+0.5,#2+0.5,#7) -- (#1-0.5,#2+0.5,#7) -- cycle;
  \draw[fill=#4!30!white,draw=black,thin] (#5,#2-0.5,#3-0.5) -- (#5,#2-0.5,#3+0.5) -- (#5,#2+0.5,#3+0.5) -- (#5,#2+0.5,#3-0.5) -- cycle;
  \draw[fill=#4!40!white,draw=black,thin] (#1-0.5,#6,#3-0.5) -- (#1+0.5,#6,#3-0.5) -- (#1+0.5,#6,#3+0.5) -- (#1-0.5,#6,#3+0.5) -- cycle;
}
\begin{tikzpicture}[math3d,scale=0.8]
  \draw[draw=black,line width=.5mm,->] (\ProjX,\ProjY,\ProjZ) -- (5,\ProjY,\ProjZ) node[left] {$x$};
  \draw[draw=black,line width=.5mm,->] (\ProjX,\ProjY,\ProjZ) -- (\ProjX,5,\ProjZ) node[right] {$y$};
  \draw[draw=black,line width=.5mm,->] (\ProjX,\ProjY,\ProjZ) -- (\ProjX,\ProjY,2) node[above] {$z$};
  \ProjectVoxel{1}{4}{0}{green}
  \ProjectVoxel{2}{2}{0}{blue}
  \ProjectVoxel{2}{3}{0}{blue}
  \ProjectVoxel{3}{1}{1}{blue}
  \ProjectVoxel{3}{2}{0}{blue}
  \ProjectVoxel{4}{0}{1}{blue}
  \ProjectVoxel{4}{0}{2}{green}
  \ProjectVoxel{4}{1}{1}{blue}
  \ProjectVoxel{2}{2}{1}{red}
  \draw[thin,draw=black] (\ProjX,2-1.5,1-1.5) -- (\ProjX,2+1.5,1-1.5) -- (\ProjX,2+1.5,1+1.5) -- (\ProjX,2-1.5,1+1.5) -- cycle;
  \draw[thin,draw=black] (2-1.5,\ProjY,1-1.5) -- (2+1.5,\ProjY,1-1.5) -- (2+1.5,\ProjY,1+1.5) -- (2-1.5,\ProjY,1+1.5) -- cycle;
  \draw[thin,draw=black] (2-1.5,2-1.5,\ProjZ) -- (2+1.5,2-1.5,\ProjZ) -- (2+1.5,2+1.5,\ProjZ) -- (2-1.5,2+1.5,\ProjZ) -- cycle;
  \DrawVoxel{1}{4}{0}{green}
  \DrawVoxel{2}{2}{0}{blue}
  \DrawVoxel{2}{3}{0}{blue}
  \DrawVoxel{3}{1}{1}{blue}
  \DrawVoxel{3}{2}{0}{blue}
  \DrawVoxel{4}{0}{1}{blue}
  \DrawVoxel{4}{0}{2}{green}
  \DrawVoxel{4}{1}{1}{blue}
  \DrawVoxel{2}{2}{1}{red}[0.6]
\end{tikzpicture}

%% file: bbox-example.tikz
\tikzset{math3d/.style={x= {(-0.5cm,-0.35cm)}, z={(0cm,0.8cm)},y={(0.7cm,-0.2cm)}}}
\NewDocumentCommand{\DrawVoxel}{m m m m O{0.75}}{%
  \draw[draw=black,thin,dotted,opacity=#5] (#1-0.5,#2-0.5,#3-0.5) -- (#1+0.5,#2-0.5,#3-0.5) -- (#1+0.5,#2+0.5,#3-0.5) -- (#1-0.5,#2+0.5,#3-0.5) -- cycle;
  \draw[draw=black,thin,dotted,opacity=#5] (#1-0.5,#2-0.5,#3-0.5) -- (#1-0.5,#2-0.5,#3+0.5) -- (#1-0.5,#2+0.5,#3+0.5) -- (#1-0.5,#2+0.5,#3-0.5) -- cycle;
  \draw[draw=black,thin,dotted,opacity=#5] (#1-0.5,#2-0.5,#3-0.5) -- (#1+0.5,#2-0.5,#3-0.5) -- (#1+0.5,#2-0.5,#3+0.5) -- (#1-0.5,#2-0.5,#3+0.5) -- cycle;
  \draw[fill=#4!10!white,draw=black,thin,opacity=#5] (#1-0.5,#2-0.5,#3+0.5) -- (#1+0.5,#2-0.5,#3+0.5) -- (#1+0.5,#2+0.5,#3+0.5) -- (#1-0.5,#2+0.5,#3+0.5) -- cycle;
  \draw[fill=#4!30!white,draw=black,thin,opacity=#5] (#1+0.5,#2-0.5,#3-0.5) -- (#1+0.5,#2-0.5,#3+0.5) -- (#1+0.5,#2+0.5,#3+0.5) -- (#1+0.5,#2+0.5,#3-0.5) -- cycle;
  \draw[fill=#4!40!white,draw=black,thin,opacity=#5] (#1-0.5,#2+0.5,#3-0.5) -- (#1+0.5,#2+0.5,#3-0.5) -- (#1+0.5,#2+0.5,#3+0.5) -- (#1-0.5,#2+0.5,#3+0.5) -- cycle;
}
\NewDocumentCommand{\DrawPoint}{m m m m O{1.0}}{%
  \draw[fill=#4,draw=black,thin,opacity=#5] (#1,#2,#3) node {\color{#4!50!white}$\bullet$};
}
  
\def\ProjX{-2.0}
\def\ProjY{-4.0}
\def\ProjZ{-2.0}
\NewDocumentCommand{\ProjectVoxel}{m m m m O{\ProjX} O{\ProjY} O{\ProjZ}}{%
  \draw[fill=#4!10!white,draw=black,thin] (#1-0.5,#2-0.5,#7) -- (#1+0.5,#2-0.5,#7) -- (#1+0.5,#2+0.5,#7) -- (#1-0.5,#2+0.5,#7) -- cycle;
  \draw[fill=#4!30!white,draw=black,thin] (#5,#2-0.5,#3-0.5) -- (#5,#2-0.5,#3+0.5) -- (#5,#2+0.5,#3+0.5) -- (#5,#2+0.5,#3-0.5) -- cycle;
  \draw[fill=#4!40!white,draw=black,thin] (#1-0.5,#6,#3-0.5) -- (#1+0.5,#6,#3-0.5) -- (#1+0.5,#6,#3+0.5) -- (#1-0.5,#6,#3+0.5) -- cycle;
}
\begin{tikzpicture}[math3d,scale=0.8]
  \draw[draw=black,line width=.5mm,->] (\ProjX,\ProjY,\ProjZ) -- (5,\ProjY,\ProjZ) node[left] {$x$};
  \draw[draw=black,line width=.5mm,->] (\ProjX,\ProjY,\ProjZ) -- (\ProjX,5,\ProjZ) node[right] {$y$};
  \draw[draw=black,line width=.5mm,->] (\ProjX,\ProjY,\ProjZ) -- (\ProjX,\ProjY,2) node[above] {$z$};
\ProjectVoxel{2}{1}{2}{blue}
\ProjectVoxel{2}{2}{2}{blue}
\ProjectVoxel{2}{0}{2}{green}
\ProjectVoxel{2}{3}{1}{blue}
\ProjectVoxel{3}{1}{2}{blue}
\ProjectVoxel{3}{2}{1}{blue}
\ProjectVoxel{3}{2}{2}{blue}
\ProjectVoxel{3}{3}{1}{blue}
\ProjectVoxel{3}{3}{2}{blue}
\ProjectVoxel{3}{4}{2}{blue}
\ProjectVoxel{3}{5}{2}{green}
\ProjectVoxel{4}{2}{1}{blue}
\ProjectVoxel{4}{2}{2}{blue}
\ProjectVoxel{4}{2}{3}{green}
\ProjectVoxel{4}{3}{1}{blue}
\ProjectVoxel{4}{3}{2}{blue}
\ProjectVoxel{4}{4}{1}{blue}
\ProjectVoxel{5}{3}{0}{green}
  
  \draw[thin,draw=black] (\ProjX,0-0.5,0-0.5) -- (\ProjX,5+0.5,0-0.5) -- (\ProjX,5+0.5,3+0.5) -- (\ProjX,0-0.5,3+0.5) -- cycle;
  \draw[thin,draw=black] (2-0.5,\ProjY,0-0.5) -- (5+0.5,\ProjY,0-0.5) -- (5+0.5,\ProjY,3+0.5) -- (2-0.5,\ProjY,3+0.5) -- cycle;
  \draw[thin,draw=black] (2-0.5,0-0.5,\ProjZ) -- (5+0.5,0-0.5,\ProjZ) -- (5+0.5,5+0.5,\ProjZ) -- (2-0.5,5+0.5,\ProjZ) -- cycle;
\DrawVoxel{2}{0}{2}{green}
\DrawVoxel{2}{1}{2}{blue}
\DrawVoxel{2}{2}{2}{blue}
\DrawVoxel{2}{3}{1}{blue}
\DrawVoxel{3}{1}{2}{blue}
\DrawVoxel{3}{2}{1}{blue}
\DrawVoxel{3}{2}{2}{blue}
\DrawVoxel{3}{3}{1}{blue}
\DrawVoxel{3}{3}{2}{blue}
\DrawVoxel{3}{4}{2}{blue}
\DrawVoxel{3}{5}{2}{green}
\DrawVoxel{4}{2}{1}{blue}
\DrawVoxel{4}{2}{2}{blue}
\DrawVoxel{4}{2}{3}{green}
\DrawVoxel{4}{3}{1}{blue}
\DrawVoxel{4}{3}{2}{blue}
\DrawVoxel{4}{4}{1}{blue}
\DrawVoxel{5}{3}{0}{green}

\end{tikzpicture}

%% file: conclu.tex
In this paper we have studied the inclusion hierarchy of full convex
sets in $\Z^d$ by examining if and how one can peel a full convex set
and keep this property throughout the process. Such sets must be peeled
necessary at vertices. Using the equivalence of full convexity with
$P$-convexity, we have exhibited easy sufficient conditions for
parallel peeling. We have then characterized the peelability
conditions through projections and proved that any full convex set is
peelable at least one point at a time, which may be found lying on its
bounding box.

Our next objective is to prove that a full convex set $X$ is peelable
under the constraint of leaving the full convex subset $Y \subset X$
unchanged. It could lead to a new method for computing the fully
convex envelope of a digital set, by Minkowski dilation, then
peeling. The current envelope operator
\cite{feschet:2022-dgmm,feschet:2023-jmiv} is indeed an iterative finite
growing process but with unclear bound: a peeling method
would be much easier to bound. A related problem is defining a
meaningful intersection for full convex sets, and peelability could
help in this endeavour.

%% file: appendix.tex
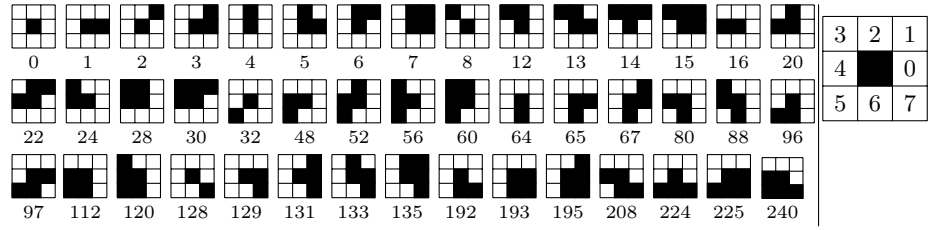
\begin{figure}
  \begin{center}
    \begin{tabular}{c|c}
    \begin{minipage}{0.87\textwidth}
      \input{peelable-cfg.tikz}
    \end{minipage}&
    \begin{tikzpicture}[scale=0.46]
      \draw[black,thin] (-1,-1) grid (2,2);
      \draw[fill=black] (0,0) rectangle ++(1,1);
      \node at (1.5,0.5)  {$0$};
      \node at (1.5,1.5)  {$1$};
      \node at (0.5,1.5)  {$2$};
      \node at (-0.5,1.5) {$3$};
      \node at (-0.5,0.5) {$4$};
      \node at (-0.5,-0.5){$5$};
      \node at (0.5,-0.5) {$6$};
      \node at (1.5,-0.5) {$7$};
    \end{tikzpicture}
    \end{tabular}
  \end{center}
  \caption{\label{fig-cfg-2d} The 45 configurations peelable in 2D. Each configuration is
    encoded as an 8-bit integer, each bit $b$ being set to 1 whenever
    the $b$-th neighbor (as displayed on the right) is present in the digital
    set. }
\end{figure}

\begin{small}
\begin{verbatim}
// Check peelability of some point in a set X, given its neighborhood
// input:  cfg  the encoded neighborhood of the point (in {0,...,255})
// output: true iff the point is peelable wrt its neighborhood
// prerequisite: the point should be a vertex of CvxH(X)
bool isLocallyPeelable( unsigned int cfg )
{ // 45 configurations are peelable
  static const unsigned int p[ 8 ]
    = { 0x5151f1ff, 0x11110001, 0x00101000b, 0x01010003,
        0x000000ab, 0x00000000, 0x00001000b, 0x00010003 };
  return p[ cfg >> 5 ] & ( 1 << (cfg & 0x1f) );
}
\end{verbatim}
\end{small}

%% file: peelable-cfg.tikz
\begin{tikzpicture}[scale=0.19]
\draw[black,thin] (-1,-1) grid (2,2);
\draw[fill=black] (0,0) rectangle ++(1,1);
\node at (0.5,-2) {\scriptsize 0};
\end{tikzpicture}
\begin{tikzpicture}[scale=0.19]
\draw[black,thin] (-1,-1) grid (2,2);
\draw[fill=black] (1,0) rectangle ++(1,1);
\draw[fill=black] (0,0) rectangle ++(1,1);
\node at (0.5,-2) {\scriptsize 1};
\end{tikzpicture}
\begin{tikzpicture}[scale=0.19]
\draw[black,thin] (-1,-1) grid (2,2);
\draw[fill=black] (1,1) rectangle ++(1,1);
\draw[fill=black] (0,0) rectangle ++(1,1);
\node at (0.5,-2) {\scriptsize 2};
\end{tikzpicture}
\begin{tikzpicture}[scale=0.19]
\draw[black,thin] (-1,-1) grid (2,2);
\draw[fill=black] (1,0) rectangle ++(1,1);
\draw[fill=black] (1,1) rectangle ++(1,1);
\draw[fill=black] (0,0) rectangle ++(1,1);
\node at (0.5,-2) {\scriptsize 3};
\end{tikzpicture}
\begin{tikzpicture}[scale=0.19]
\draw[black,thin] (-1,-1) grid (2,2);
\draw[fill=black] (0,1) rectangle ++(1,1);
\draw[fill=black] (0,0) rectangle ++(1,1);
\node at (0.5,-2) {\scriptsize 4};
\end{tikzpicture}
\begin{tikzpicture}[scale=0.19]
\draw[black,thin] (-1,-1) grid (2,2);
\draw[fill=black] (1,0) rectangle ++(1,1);
\draw[fill=black] (0,1) rectangle ++(1,1);
\draw[fill=black] (0,0) rectangle ++(1,1);
\node at (0.5,-2) {\scriptsize 5};
\end{tikzpicture}
\begin{tikzpicture}[scale=0.19]
\draw[black,thin] (-1,-1) grid (2,2);
\draw[fill=black] (1,1) rectangle ++(1,1);
\draw[fill=black] (0,1) rectangle ++(1,1);
\draw[fill=black] (0,0) rectangle ++(1,1);
\node at (0.5,-2) {\scriptsize 6};
\end{tikzpicture}
\begin{tikzpicture}[scale=0.19]
\draw[black,thin] (-1,-1) grid (2,2);
\draw[fill=black] (1,0) rectangle ++(1,1);
\draw[fill=black] (1,1) rectangle ++(1,1);
\draw[fill=black] (0,1) rectangle ++(1,1);
\draw[fill=black] (0,0) rectangle ++(1,1);
\node at (0.5,-2) {\scriptsize 7};
\end{tikzpicture}
\begin{tikzpicture}[scale=0.19]
\draw[black,thin] (-1,-1) grid (2,2);
\draw[fill=black] (-1,1) rectangle ++(1,1);
\draw[fill=black] (0,0) rectangle ++(1,1);
\node at (0.5,-2) {\scriptsize 8};
\end{tikzpicture}
\begin{tikzpicture}[scale=0.19]
\draw[black,thin] (-1,-1) grid (2,2);
\draw[fill=black] (0,1) rectangle ++(1,1);
\draw[fill=black] (-1,1) rectangle ++(1,1);
\draw[fill=black] (0,0) rectangle ++(1,1);
\node at (0.5,-2) {\scriptsize 12};
\end{tikzpicture}
\begin{tikzpicture}[scale=0.19]
\draw[black,thin] (-1,-1) grid (2,2);
\draw[fill=black] (1,0) rectangle ++(1,1);
\draw[fill=black] (0,1) rectangle ++(1,1);
\draw[fill=black] (-1,1) rectangle ++(1,1);
\draw[fill=black] (0,0) rectangle ++(1,1);
\node at (0.5,-2) {\scriptsize 13};
\end{tikzpicture}
\begin{tikzpicture}[scale=0.19]
\draw[black,thin] (-1,-1) grid (2,2);
\draw[fill=black] (1,1) rectangle ++(1,1);
\draw[fill=black] (0,1) rectangle ++(1,1);
\draw[fill=black] (-1,1) rectangle ++(1,1);
\draw[fill=black] (0,0) rectangle ++(1,1);
\node at (0.5,-2) {\scriptsize 14};
\end{tikzpicture}
\begin{tikzpicture}[scale=0.19]
\draw[black,thin] (-1,-1) grid (2,2);
\draw[fill=black] (1,0) rectangle ++(1,1);
\draw[fill=black] (1,1) rectangle ++(1,1);
\draw[fill=black] (0,1) rectangle ++(1,1);
\draw[fill=black] (-1,1) rectangle ++(1,1);
\draw[fill=black] (0,0) rectangle ++(1,1);
\node at (0.5,-2) {\scriptsize 15};
\end{tikzpicture}
\begin{tikzpicture}[scale=0.19]
\draw[black,thin] (-1,-1) grid (2,2);
\draw[fill=black] (-1,0) rectangle ++(1,1);
\draw[fill=black] (0,0) rectangle ++(1,1);
\node at (0.5,-2) {\scriptsize 16};
\end{tikzpicture}
\begin{tikzpicture}[scale=0.19]
\draw[black,thin] (-1,-1) grid (2,2);
\draw[fill=black] (0,1) rectangle ++(1,1);
\draw[fill=black] (-1,0) rectangle ++(1,1);
\draw[fill=black] (0,0) rectangle ++(1,1);
\node at (0.5,-2) {\scriptsize 20};
\end{tikzpicture}
\begin{tikzpicture}[scale=0.19]
\draw[black,thin] (-1,-1) grid (2,2);
\draw[fill=black] (1,1) rectangle ++(1,1);
\draw[fill=black] (0,1) rectangle ++(1,1);
\draw[fill=black] (-1,0) rectangle ++(1,1);
\draw[fill=black] (0,0) rectangle ++(1,1);
\node at (0.5,-2) {\scriptsize 22};
\end{tikzpicture}
\begin{tikzpicture}[scale=0.19]
\draw[black,thin] (-1,-1) grid (2,2);
\draw[fill=black] (-1,1) rectangle ++(1,1);
\draw[fill=black] (-1,0) rectangle ++(1,1);
\draw[fill=black] (0,0) rectangle ++(1,1);
\node at (0.5,-2) {\scriptsize 24};
\end{tikzpicture}
\begin{tikzpicture}[scale=0.19]
\draw[black,thin] (-1,-1) grid (2,2);
\draw[fill=black] (0,1) rectangle ++(1,1);
\draw[fill=black] (-1,1) rectangle ++(1,1);
\draw[fill=black] (-1,0) rectangle ++(1,1);
\draw[fill=black] (0,0) rectangle ++(1,1);
\node at (0.5,-2) {\scriptsize 28};
\end{tikzpicture}
\begin{tikzpicture}[scale=0.19]
\draw[black,thin] (-1,-1) grid (2,2);
\draw[fill=black] (1,1) rectangle ++(1,1);
\draw[fill=black] (0,1) rectangle ++(1,1);
\draw[fill=black] (-1,1) rectangle ++(1,1);
\draw[fill=black] (-1,0) rectangle ++(1,1);
\draw[fill=black] (0,0) rectangle ++(1,1);
\node at (0.5,-2) {\scriptsize 30};
\end{tikzpicture}
\begin{tikzpicture}[scale=0.19]
\draw[black,thin] (-1,-1) grid (2,2);
\draw[fill=black] (-1,-1) rectangle ++(1,1);
\draw[fill=black] (0,0) rectangle ++(1,1);
\node at (0.5,-2) {\scriptsize 32};
\end{tikzpicture}
\begin{tikzpicture}[scale=0.19]
\draw[black,thin] (-1,-1) grid (2,2);
\draw[fill=black] (-1,0) rectangle ++(1,1);
\draw[fill=black] (-1,-1) rectangle ++(1,1);
\draw[fill=black] (0,0) rectangle ++(1,1);
\node at (0.5,-2) {\scriptsize 48};
\end{tikzpicture}
\begin{tikzpicture}[scale=0.19]
\draw[black,thin] (-1,-1) grid (2,2);
\draw[fill=black] (0,1) rectangle ++(1,1);
\draw[fill=black] (-1,0) rectangle ++(1,1);
\draw[fill=black] (-1,-1) rectangle ++(1,1);
\draw[fill=black] (0,0) rectangle ++(1,1);
\node at (0.5,-2) {\scriptsize 52};
\end{tikzpicture}
\begin{tikzpicture}[scale=0.19]
\draw[black,thin] (-1,-1) grid (2,2);
\draw[fill=black] (-1,1) rectangle ++(1,1);
\draw[fill=black] (-1,0) rectangle ++(1,1);
\draw[fill=black] (-1,-1) rectangle ++(1,1);
\draw[fill=black] (0,0) rectangle ++(1,1);
\node at (0.5,-2) {\scriptsize 56};
\end{tikzpicture}
\begin{tikzpicture}[scale=0.19]
\draw[black,thin] (-1,-1) grid (2,2);
\draw[fill=black] (0,1) rectangle ++(1,1);
\draw[fill=black] (-1,1) rectangle ++(1,1);
\draw[fill=black] (-1,0) rectangle ++(1,1);
\draw[fill=black] (-1,-1) rectangle ++(1,1);
\draw[fill=black] (0,0) rectangle ++(1,1);
\node at (0.5,-2) {\scriptsize 60};
\end{tikzpicture}
\begin{tikzpicture}[scale=0.19]
\draw[black,thin] (-1,-1) grid (2,2);
\draw[fill=black] (0,-1) rectangle ++(1,1);
\draw[fill=black] (0,0) rectangle ++(1,1);
\node at (0.5,-2) {\scriptsize 64};
\end{tikzpicture}
\begin{tikzpicture}[scale=0.19]
\draw[black,thin] (-1,-1) grid (2,2);
\draw[fill=black] (1,0) rectangle ++(1,1);
\draw[fill=black] (0,-1) rectangle ++(1,1);
\draw[fill=black] (0,0) rectangle ++(1,1);
\node at (0.5,-2) {\scriptsize 65};
\end{tikzpicture}
\begin{tikzpicture}[scale=0.19]
\draw[black,thin] (-1,-1) grid (2,2);
\draw[fill=black] (1,0) rectangle ++(1,1);
\draw[fill=black] (1,1) rectangle ++(1,1);
\draw[fill=black] (0,-1) rectangle ++(1,1);
\draw[fill=black] (0,0) rectangle ++(1,1);
\node at (0.5,-2) {\scriptsize 67};
\end{tikzpicture}
\begin{tikzpicture}[scale=0.19]
\draw[black,thin] (-1,-1) grid (2,2);
\draw[fill=black] (-1,0) rectangle ++(1,1);
\draw[fill=black] (0,-1) rectangle ++(1,1);
\draw[fill=black] (0,0) rectangle ++(1,1);
\node at (0.5,-2) {\scriptsize 80};
\end{tikzpicture}
\begin{tikzpicture}[scale=0.19]
\draw[black,thin] (-1,-1) grid (2,2);
\draw[fill=black] (-1,1) rectangle ++(1,1);
\draw[fill=black] (-1,0) rectangle ++(1,1);
\draw[fill=black] (0,-1) rectangle ++(1,1);
\draw[fill=black] (0,0) rectangle ++(1,1);
\node at (0.5,-2) {\scriptsize 88};
\end{tikzpicture}
\begin{tikzpicture}[scale=0.19]
\draw[black,thin] (-1,-1) grid (2,2);
\draw[fill=black] (-1,-1) rectangle ++(1,1);
\draw[fill=black] (0,-1) rectangle ++(1,1);
\draw[fill=black] (0,0) rectangle ++(1,1);
\node at (0.5,-2) {\scriptsize 96};
\end{tikzpicture}
\begin{tikzpicture}[scale=0.19]
\draw[black,thin] (-1,-1) grid (2,2);
\draw[fill=black] (1,0) rectangle ++(1,1);
\draw[fill=black] (-1,-1) rectangle ++(1,1);
\draw[fill=black] (0,-1) rectangle ++(1,1);
\draw[fill=black] (0,0) rectangle ++(1,1);
\node at (0.5,-2) {\scriptsize 97};
\end{tikzpicture}\hfill
\begin{tikzpicture}[scale=0.19]
\draw[black,thin] (-1,-1) grid (2,2);
\draw[fill=black] (-1,0) rectangle ++(1,1);
\draw[fill=black] (-1,-1) rectangle ++(1,1);
\draw[fill=black] (0,-1) rectangle ++(1,1);
\draw[fill=black] (0,0) rectangle ++(1,1);
\node at (0.5,-2) {\scriptsize 112};
\end{tikzpicture}\hfill
\begin{tikzpicture}[scale=0.19]
\draw[black,thin] (-1,-1) grid (2,2);
\draw[fill=black] (-1,1) rectangle ++(1,1);
\draw[fill=black] (-1,0) rectangle ++(1,1);
\draw[fill=black] (-1,-1) rectangle ++(1,1);
\draw[fill=black] (0,-1) rectangle ++(1,1);
\draw[fill=black] (0,0) rectangle ++(1,1);
\node at (0.5,-2) {\scriptsize 120};
\end{tikzpicture}\hfill
\begin{tikzpicture}[scale=0.19]
\draw[black,thin] (-1,-1) grid (2,2);
\draw[fill=black] (1,-1) rectangle ++(1,1);
\draw[fill=black] (0,0) rectangle ++(1,1);
\node at (0.5,-2) {\scriptsize 128};
\end{tikzpicture}\hfill
\begin{tikzpicture}[scale=0.19]
\draw[black,thin] (-1,-1) grid (2,2);
\draw[fill=black] (1,0) rectangle ++(1,1);
\draw[fill=black] (1,-1) rectangle ++(1,1);
\draw[fill=black] (0,0) rectangle ++(1,1);
\node at (0.5,-2) {\scriptsize 129};
\end{tikzpicture}\hfill
\begin{tikzpicture}[scale=0.19]
\draw[black,thin] (-1,-1) grid (2,2);
\draw[fill=black] (1,0) rectangle ++(1,1);
\draw[fill=black] (1,1) rectangle ++(1,1);
\draw[fill=black] (1,-1) rectangle ++(1,1);
\draw[fill=black] (0,0) rectangle ++(1,1);
\node at (0.5,-2) {\scriptsize 131};
\end{tikzpicture}\hfill
\begin{tikzpicture}[scale=0.19]
\draw[black,thin] (-1,-1) grid (2,2);
\draw[fill=black] (1,0) rectangle ++(1,1);
\draw[fill=black] (0,1) rectangle ++(1,1);
\draw[fill=black] (1,-1) rectangle ++(1,1);
\draw[fill=black] (0,0) rectangle ++(1,1);
\node at (0.5,-2) {\scriptsize 133};
\end{tikzpicture}\hfill
\begin{tikzpicture}[scale=0.19]
\draw[black,thin] (-1,-1) grid (2,2);
\draw[fill=black] (1,0) rectangle ++(1,1);
\draw[fill=black] (1,1) rectangle ++(1,1);
\draw[fill=black] (0,1) rectangle ++(1,1);
\draw[fill=black] (1,-1) rectangle ++(1,1);
\draw[fill=black] (0,0) rectangle ++(1,1);
\node at (0.5,-2) {\scriptsize 135};
\end{tikzpicture}\hfill
\begin{tikzpicture}[scale=0.19]
\draw[black,thin] (-1,-1) grid (2,2);
\draw[fill=black] (0,-1) rectangle ++(1,1);
\draw[fill=black] (1,-1) rectangle ++(1,1);
\draw[fill=black] (0,0) rectangle ++(1,1);
\node at (0.5,-2) {\scriptsize 192};
\end{tikzpicture}\hfill
\begin{tikzpicture}[scale=0.19]
\draw[black,thin] (-1,-1) grid (2,2);
\draw[fill=black] (1,0) rectangle ++(1,1);
\draw[fill=black] (0,-1) rectangle ++(1,1);
\draw[fill=black] (1,-1) rectangle ++(1,1);
\draw[fill=black] (0,0) rectangle ++(1,1);
\node at (0.5,-2) {\scriptsize 193};
\end{tikzpicture}\hfill
\begin{tikzpicture}[scale=0.19]
\draw[black,thin] (-1,-1) grid (2,2);
\draw[fill=black] (1,0) rectangle ++(1,1);
\draw[fill=black] (1,1) rectangle ++(1,1);
\draw[fill=black] (0,-1) rectangle ++(1,1);
\draw[fill=black] (1,-1) rectangle ++(1,1);
\draw[fill=black] (0,0) rectangle ++(1,1);
\node at (0.5,-2) {\scriptsize 195};
\end{tikzpicture}\hfill
\begin{tikzpicture}[scale=0.19]
\draw[black,thin] (-1,-1) grid (2,2);
\draw[fill=black] (-1,0) rectangle ++(1,1);
\draw[fill=black] (0,-1) rectangle ++(1,1);
\draw[fill=black] (1,-1) rectangle ++(1,1);
\draw[fill=black] (0,0) rectangle ++(1,1);
\node at (0.5,-2) {\scriptsize 208};
\end{tikzpicture}\hfill
\begin{tikzpicture}[scale=0.19]
\draw[black,thin] (-1,-1) grid (2,2);
\draw[fill=black] (-1,-1) rectangle ++(1,1);
\draw[fill=black] (0,-1) rectangle ++(1,1);
\draw[fill=black] (1,-1) rectangle ++(1,1);
\draw[fill=black] (0,0) rectangle ++(1,1);
\node at (0.5,-2) {\scriptsize 224};
\end{tikzpicture}\hfill
\begin{tikzpicture}[scale=0.19]
\draw[black,thin] (-1,-1) grid (2,2);
\draw[fill=black] (1,0) rectangle ++(1,1);
\draw[fill=black] (-1,-1) rectangle ++(1,1);
\draw[fill=black] (0,-1) rectangle ++(1,1);
\draw[fill=black] (1,-1) rectangle ++(1,1);
\draw[fill=black] (0,0) rectangle ++(1,1);
\node at (0.5,-2) {\scriptsize 225};
\end{tikzpicture}\hfill
\begin{tikzpicture}[scale=0.18]
\draw[black,thin] (-1,-1) grid (2,2);
\draw[fill=black] (-1,0) rectangle ++(1,1);
\draw[fill=black] (-1,-1) rectangle ++(1,1);
\draw[fill=black] (0,-1) rectangle ++(1,1);
\draw[fill=black] (1,-1) rectangle ++(1,1);
\draw[fill=black] (0,0) rectangle ++(1,1);
\node at (0.5,-2) {\scriptsize 240};
\end{tikzpicture}

%% file: paper.bbl
\begin{thebibliography}{10}
\providecommand{\url}[1]{\texttt{#1}}
\providecommand{\urlprefix}{URL }
\providecommand{\doi}[1]{https://doi.org/#1}

\bibitem{bertrand2014powerful}
Bertrand, G., Couprie, M.: Powerful parallel and symmetric 3d thinning schemes
  based on critical kernels. Journal of Mathematical Imaging and Vision
  \textbf{48},  134--148 (2014)

\bibitem{bertrand1994new}
Bertrand, G., Malandain, G.: A new characterization of three-dimensional simple
  points. Pattern Recognition Letters  \textbf{15}(2),  169--175 (1994)

\bibitem{brlek2009lyndon+}
Brlek, S., Lachaud, J.O., Proven{\c{c}}al, X., Reutenauer, C.: Lyndon+
  christoffel= digitally convex. Pattern Recognition  \textbf{42}(10),
  2239--2246 (2009)

\bibitem{couprie2008new}
Couprie, M., Bertrand, G.: New characterizations of simple points in 2d, 3d,
  and 4d discrete spaces. IEEE Transactions on Pattern Analysis and Machine
  Intelligence  \textbf{31}(4),  637--648 (2008)

\bibitem{couprie2016asymmetric}
Couprie, M., Bertrand, G.: Asymmetric parallel 3d thinning scheme and
  algorithms based on isthmuses. Pattern Recognition Letters  \textbf{76},
  22--31 (2016)

\bibitem{Eckhardt2000}
Eckhardt, U.: Digital lines and digital convexity. In: Digital and Image
  Geometry, Advanced Lectures. LNCS, vol.~2243, p. 209–228 (2001)

\bibitem{feschet:2022-dgmm}
Feschet, F., Lachaud, J.O.: Full convexity for polyhedral models in digital
  spaces. In: Discrete Geometry and Mathematical Morphology. pp. 98--109.
  Springer International Publishing, Cham (2022)

\bibitem{feschet:2023-jmiv}
Feschet, F., Lachaud, J.O.: An envelope operator for full convexity to define
  polyhedral models in digital spaces. J. Math. Imaging Vis.  \textbf{65}(5),
  754--769 (2023). \doi{10.1007/s10851-023-01155-w}

\bibitem{feschet:2024-dgmm}
Feschet, F., Lachaud, J.O.: New characterizations of full convexity. In:
  Brunetti, S., Frosini, A., Rinaldi, S. (eds.) Discrete Geometry and
  Mathematical Morphology. pp. 41--53. Springer Nature Switzerland, Cham (2024)

\bibitem{feschet:2025-jmiv}
Feschet, F., Lachaud, J.O.: New properties for full convex sets and full convex
  hulls. J. Math. Imaging Vis.  \textbf{67}(8) (2025).
  \doi{10.1007/s10851-024-01225-7}

\bibitem{kim:1982-tpami-3d}
Kim, C.E., Rosenfeld, A.: Convex digital solids. IEEE Trans. Pattern Anal.
  Machine Intel.  \textbf{6},  612--618 (1982)

\bibitem{kim:1982-tpami}
Kim, C.E., Rosenfeld, A.: Digital straight lines and convexity of digital
  regions. IEEE Trans. Pattern Anal. Machine Intel.  \textbf{2},  149--153
  (1982)

\bibitem{kim:1982-pr}
Kim, C.E., Sklansky, J.: Digital and cellular convexity. Pattern Recognition
  \textbf{15}(5),  359 -- 367 (1982). \doi{10.1016/0031-3203(82)90038-3}

\bibitem{Kiselman2004}
Kiselman, C.O.: Convex functions on discrete sets. In: Combinatorial Image
  Analysis. 10th International Workshop, IWCIA. pp. 443--457. LNCS 3322,
  Springer-Verlag (2004)

\bibitem{kiselman:2022}
Kiselman, C.O.: Elements of Digital Geometry, Mathematical Morphology, and
  Discrete Optimization. World Scientific (2022). \doi{10.1142/12584}

\bibitem{lachaud:2021-dgmm}
Lachaud, J.O.: An alternative definition for digital convexity. In: Discrete
  Geometry and Mathematical Morphology (DGMM'2021), Uppsala, Sweden, Proc.
  LNCS, vol. 12708, pp. 269--282. Springer (2021).
  \doi{10.1007/978-3-030-76657-3\_19}

\bibitem{lachaud:2022-jmiv}
Lachaud, J.O.: An alternative definition for digital convexity. J. Math.
  Imaging Vis.  \textbf{64}(7),  718--735 (2022).
  \doi{10.1007/s10851-022-01076-0}

\bibitem{Llinares2002}
Llinares, J.V.: Abstract convexity, some relations and applications.
  Optimization  \textbf{51}(6),  797--818 (2002)

\bibitem{Murota2001}
Murota, K., Shioura, A.: Relationship of m/l-convex functions with discrete
  convex functions by {M}iller and {F}avati–{T}ardella. Discrete Applied
  Maths  \textbf{115},  151--176 (2001)

\bibitem{pudney1998distance}
Pudney, C.: Distance-ordered homotopic thinning: a skeletonization algorithm
  for 3d digital images. Computer vision and image understanding
  \textbf{72}(3),  404--413 (1998)

\bibitem{ronse:1989-tpami}
Ronse, C.: A bibliography on digital and computational convexity (1961-1988).
  IEEE Trans. Pattern Anal. Machine Intel.  \textbf{11}(2),  181--190 (1989)

\bibitem{RoyStell2003}
Roy, A.J., Stell, J.G.: Convexity in discrete space. In: Kuhn, W., Worboys, M.,
  Timpf, S. (eds.) COSIT. pp. 253--269. LNCS 2825, Springer-Verlag (2003)

\bibitem{tarsissi2019convexity}
Tarsissi, L., Coeurjolly, D., Kenmochi, Y., Romon, P.: Convexity preserving
  contraction of digital sets. In: Asian Conference on Pattern Recognition. pp.
  611--624. Springer (2019)

\bibitem{tarsissi2022algorithms}
Tarsissi, L., Kenmochi, Y., Djerroumi, H., Coeurjolly, D., Romon, P., Borel,
  J.P.: Algorithms for pixelwise shape deformations preserving digital
  convexity. In: International Conference on Discrete Geometry and Mathematical
  Morphology. pp. 84--97. Springer (2022)

\bibitem{tarsissi2023convexity}
Tarsissi, L., Kenmochi, Y., Romon, P., Coeurjolly, D., Borel, J.P.: Convexity
  preserving deformations of digital sets: Characterization of removable and
  insertable pixels. Discrete Applied Mathematics  \textbf{341},  270--289
  (2023)

\bibitem{webster2001cell}
Webster, J.: Cell complexes and digital convexity. In: Digital and Image
  Geometry: Advanced Lectures. LNCS, vol.~2243, pp. 272--282 (2001).
  \doi{10.1007/3-540-45576-0_16}

\end{thebibliography}
